\documentclass[sort&compress,number]{elsarticle}

\usepackage{amsmath}
\usepackage{amssymb}
\usepackage{graphics} 
 \usepackage{etoolbox}
\usepackage{setspace}
\usepackage{mathtools}
\usepackage{amsfonts}
\usepackage{pdfpages}
\usepackage{epstopdf}
\usepackage{subfig}
\usepackage{epsfig} 
\usepackage{enumerate}
\usepackage{algorithm}
\usepackage{algpseudocode}
\def\+{\!+\!}
\def\-{\!-\!}

\def\itbf#1{\textit{\textbf{#1}}}

\def\notreg(#1,#2){$#2$ is not $#1$-regular}

\def\LCAF{{\tt{LCAF}}}

\def\pref(#1,#2){$#1$ is a prefix of $#2$}

\def\reg(#1,#2){$#2$ is $#1$-regular}

\def\s#1{\mbox{\boldmath $#1$}}
\def\suff(#1,#2){$#1$ is a suffix of $#2$}

\def\UPDATE\_F{\tt{UPDATE\_F}}

\newcommand{\BO}{\mathcal{O}}

\def\itbf#1{\textit{\textbf{#1}}}

\title{Algorithms for Longest Common Abelian Factors}
\author[kcl]{Ali Alatabbi}
\ead{Ali.Alatabbi@kcl.ac.uk}
\author[kcl]{Costas S.~Iliopoulos}
\ead{Costas.Iliopoulos@kcl.ac.uk}
\author[kcl,iamc]{Alessio Langiu\corref{cor1}}
\ead{Alessio.Langiu@kcl.ac.uk}
\author[kcl,buet]{M. Sohel Rahman}
\ead{msrahman@cse.buet.ac.bd}
\address[kcl]{Department of Informatics, King's College London, London, UK}
\address[iamc]{IAMC-CNR, National Research Council,Trapani, Italy}
\address[buet]{A$\ell$EDA Group, Department of CSE, BUET, Dhaka-1000, Bangladesh}
\cortext[cor1]{Corresponding author}

\newtheorem{thm}{Theorem}
\newtheorem{problem}[thm]{Problem}
\newtheorem{lemma}[thm]{Lemma}
\newtheorem{conjecture}[thm]{Conjecture}
\newdefinition{definition}{Definition}
\newproof{proof}{Proof}

\newcommand{\dd}{\mathinner{\ldotp\ldotp}}

\newcommand{\pv}{Parikh vector}

\begin{document}

\begin{abstract}
In this paper we consider the problem of computing the longest common abelian
factor (LCAF) between two given strings. We present a simple $\BO(\sigma~ n^2)$
time algorithm, where $n$ is the length of the strings and $\sigma$ is the alphabet size, and a sub-quadratic running time solution for the binary string case, both having linear space requirement. Furthermore, we present a modified algorithm applying some interesting tricks and experimentally show that the resulting algorithm runs faster.
\end{abstract}

\maketitle 

\section{\label{intro}Introduction}

Abelian properties concerning words have been investigated since
the very beginning of the study of Formal Languages and Combinatorics on Words.
Abelian powers were first considered in 1961 by Erd\H{o}s \cite{Erdos1961221} as a natural generalization of usual powers. In 1966, Parikh \cite{Parikh:1966:CL:321356.321364} defined a vector having length equal to the alphabet cardinality, which reports the number of occurrences of each alphabet symbol inside a given string. Later on, the scientific community started referring to such a vector as the \emph{\pv}. Clearly, two strings having the same \pv~ are permutations of one another and there is an \emph{abelian match} between them.

Abelian properties of strings have recently grown tremendous interest among the Stringology researchers and have become an involving topic of discussion in the recent issues of the StringMasters meetings. Despite the fact that there are not so many real life applications where comparing commutative sequence of objects is relevant, abelian combinatorics has a potential role in filtering the data in order to find potential occurrences of some approximate matches. For instance, when one is looking for typing errors in a natural language, it can be useful to select the abelian matches first and then look for swap of adjacent or even near appearing letters. The swap errors and the inversion errors are also very common in the evolutionary process of the genetic code of a living organism and hence is often interesting from Bioinformatics perspective. Similar applications can also be found in the context of network communications.

In this paper, we focus on the problem of finding the Longest Common Abelian Factor of two given strings. The problem is combinatorially interesting and analogous to the Longest Common Substring (LCStr) problem for the usual strings. The LCStr problem is a Historical problem and Dan Gusfield reported the following in his book \cite[Sec. 7.4]{DBLP:books/cu/Gusfield1997} regarding the belief of Don Knuth about the complexity of the problem:
\smallskip
\begin{quote}
\indent ...in 1970 Don Knuth conjectured a linear time algorithm for this problem would be impossible.
\end{quote}
\smallskip
 However, contrary to the above conjecture, decades later, a linear time
 solution for the LCStr problem was in fact obtained by using the linear
 construction of the suffix tree. For Stringology researchers this alone could
 be the motivation for considering LCAF from both algorithmic and combinatorics
 point of view. However, despite a number of works on abelian matching, to the
 best of our knowledge, this problem has never been considered until very recently when it was posed in the latest issue of the StringMasters, i.e., StringMasters 2013. To this end, this research work can be seen as a first attempt to solve this problem with the hope of many more to follow.

 In this paper, we first present a simple solution to the problem running in
 $\BO(\sigma~ n^2)$ time, where $\sigma$ is the alphabet size (Section
 \ref{SEC:quadratic}). Then we present a sub-quadratic algorithm for the binary
 string case (Section \ref{SEC:subquadratic}). Both the algorithms have linear
 space requirement. Furthermore, we present a modified algorithm applying some
 interesting tricks (Section \ref{SEC:avg}) and experimentally show that the
 resulting algorithm runs in $\BO(n \log n)$ time (Section \ref{sec:exp}).

\section{Preliminaries}\label{SEC:definitions}
An alphabet $\Sigma$ of size $\sigma>0$ is a finite set whose elements are called letters. A string on an alphabet $\Sigma$ is a finite, possibly empty, sequence of elements of $\Sigma$. The zero-letter sequence is called the empty string, and is denoted by $\varepsilon$. The length of a string $S$ is defined as the length of the sequence associated with the string $S$, and is denoted by $|S|$. We denote by $S[i]$ the $i$-th letter of $S$, for
all $1 \leq i \leq |S|$ and $S=S[1 \dd |S|]$.
A string $w$ is a factor of a string $S$ if there exist two strings $u$ and $v$, possibly empty, such that $S =uwv$. A factor $w$ of a string $S$ is proper if $w \neq S$. If $u =\varepsilon$ ($v =\varepsilon$), then $w$ is a prefix (suffix) of $S$.

Given a string $S$ over the alphabet $\Sigma = \{a_1, \ldots a_{\sigma}\}$, we denote by $|S|_{a_{j}}$ the number of $a_{j}$'s in $S$, for $1 \leq j \leq \sigma$.
We define the \pv ~of $S$ as $\mathcal {P}_S=(|S|_{a_{1}}, \ldots
|S|_{a_\sigma})$.

In the binary case, we denote $\Sigma = \{0,1\}$, the number of $0$'s in $S$ by $|S|_0$,  the number of $1$'s in $S$ by $|S|_1$ and the \pv ~of $S$ as $\mathcal{P}_S=(|S|_0,|S|_1)$. We now focus on binary strings. The general alphabet case will be considered later.

For a given binary string $S$ of length $n$, we define an $n\times n$ matrix $M_S$ as follows. Each row of $M_S$ is dedicated to a particular length of factors of $S$.
So, Row $\ell$ of $M_S$ is dedicated to $\ell$-length factors of $S$.
Each column of $M_S$  is dedicated to a particular starting position of factors of $S$. So, Column $i$ of $M_S$  is dedicated to the position $i$ of $S$.
Hence, $M_S[\ell][i]$ is dedicated to the $\ell$-length factor that starts at position $i$ of $S$ and it reports the number of $1$'s of that factor.
Now, $M_S[\ell][i] = m$ if and only if the $\ell$-length factor that starts at position $i$ of $S$ has a total of $m$ $1$'s, that is, $|S[i \dd i+\ell-1]|_1 = m$. We formally define the matrix $M_S$ as follows.

\begin{definition}
Given a binary string $S$ of length $n$, $M_S$ is an $n\times n$ matrix such that $M_S[\ell][i] = |S[i \dd i+\ell-1]|_1$, for $1 \leq \ell \leq n$ and $1 \leq i \leq (n-\ell +1)$, and $M_S[\ell][i] = 0$, otherwise.
\end{definition}

In what follows, we will use $M_S[\ell]$ to refer to Row $\ell$ of $M_S$. Assume that we are given two strings $A$ and $B$ on an alphabet $\Sigma$. 
For the sake of ease, we assume that $|A| = |B| = n$.
We want to find the length of a longest common abelian factor between $A$ and $B$.

\begin{definition}
Given two strings $A$ and $B$ over the alphabet $\Sigma$, we say that $w$ is a common abelian factor for $A$ and $B$ if there exist a factor (or substring) $u$ in $A$ and a factor $v$ in $B$ such that $\mathcal{P}_w = \mathcal{P}_u = \mathcal{P}_v$. A common abelian factor of the highest length is called the \emph{Longest Common Abelian Factor (LCAF)} between $A$ and $B$. The length of LCAF is referred to as the \emph{LCAF length}.
\end{definition}
In this paper we study the following problem.

\begin{problem}[LCAF Problem]
Given two strings $A$ and $B$ over the alphabet $\Sigma$, compute the length of an LCAF and identify some occurrences of an LCAF between $A$ and $B$ .
\end{problem}




Assume that the strings $A$ and $B$ of length $n$ are given.
%
%
Now, suppose that the matrices $M_A$ and $M_B$ for the binary strings $A$ and
$B$ have been computed. Now we have the following easy lemma that will be useful
for us later.

\begin{lemma}\label{lem:factor}
 There is a common abelian factor of length $\ell$ between $A$ and $B$ if and only if there exists $p,q$ such that $1\leq p,q \leq n - \ell + 1$ and $M_A[\ell][p] = M_B[\ell][q]$.
\end{lemma}
\begin{proof}
Suppose there exists $p,q$ such that $1\leq p,q \leq n - \ell + 1$ and $M_A[\ell][p] = M_B[\ell][q]$. By definition this means $|A[p \dd p+\ell-1]|_1 = |B[q \dd q+\ell-1]|_1$. So there is a common abelian factor of length $\ell$ between $A$ and $B$.
The other way is also obvious by definition. 
\end{proof}

Clearly, if we have $M_A$ and $M_B$ we can compute the LCAF by identifying the highest $\ell$ such that there exists $p,q$ having $1\leq p,q \leq n - \ell + 1$ and $M_A[\ell][p] = M_B[\ell][q]$.
Then we can say that the LCAF between $A$ and $B$ is either $A[p \dd p+\ell-1]$ or $B[q\dd q+\ell-1]$ having length $\ell$. 

We now generalize the definition of the matrix $M_S$ for strings over a fixed size alphabet $\Sigma = \{a_1, \ldots a_\sigma \}$ by  defining an $n\times n$ matrix $M_S$ of $(\sigma - 1)$-length vectors. $M_S[\ell][i] = V_{\ell,i}$, where $V_{\ell,i}[j] = |S[i \dd i+\ell-1]|_{a_j}$, for $1 \leq \ell \leq n$, $1 \leq i \leq (n-\ell +1)$ and $1 \leq j < \sigma$, and $V_{\ell,i}[j] = 0$, otherwise. We will refer to the $j$-th element of the array $V_{\ell,i}$ of the matrix $M_S$ by using the notation $M_S[\ell][i][j]$. Notice that the last component of a \pv ~is determined by using the length of the string and all the other components of the \pv .
Now, $M_S[\ell][i][j] = m$ if and only if the $\ell$-length factor that starts at position $i$ of $S$ has a total of $m$ $a_j$'s, that is $|S[i \dd i+\ell-1]|_{a_j} = m$. Clearly, we can compute $M_S[\ell]$ using the following steps.
Similar to the binary case, the above computation runs in linear time because we
can compute $|S[i+1 \dd i+1+\ell -1]|_{a_j}$ from $|S[i \dd i+\ell -1]|_{a_j}$
in constant time by simply decrementing the $S[i]$ component and incrementing the $S[i +\ell]$ one.

\section{A Quadratic Algorithm}\label{SEC:quadratic}

A simple approach for finding the LCAF length considers computing, for $1\leq \ell\leq n$, the \pv s of all the factors of length $\ell$ in both $A$ and $B$, i.e., $M_A[\ell]$ and $M_B[\ell]$. Then, we check whether $M_A[\ell]$ and $M_B[\ell]$ have non-empty intersection. If yes, then $\ell$ could be the LCAF length. So, we return the highest of such $\ell$.
%
 Moreover, if one knows a \pv ~having the LCAF length belonging to such
 intersection, a linear scan of $A$ and $B$ produces one occurrence of such a
 factor. The asymptotic time complexity of this approach is $\BO(\sigma ~n^2 )$
 and it requires $\BO(\sigma ~n \log n)$ bits of extra space. The basic steps
 are outlined as follows.
%
%

%
\begin{enumerate}
  \item For $\ell = 1$ to $n$ do the following
  \item \ \ \ For $i = 1$ to $n - \ell + 1$ do the following
  \item \ \ \ \ \ \ compute $M_A[\ell][i]$ and $M_B[\ell][i]$
  \item \ \ \ If $M_A[\ell] \bigcap M_B[\ell] \neq\emptyset$ then
  \item \ \ \ \ \ \ LCAF $= \ell$
\end{enumerate}

It is well known that, for fixed length $\ell$, one can compute all the \pv s in
linear time and store them in $\BO(\sigma ~n \log n)$ bits. 
Now once $M_A$ and $M_B$ are computed, we simply need to apply the idea of Lemma~\ref{lem:factor}. The idea is to check for all values of $\ell$ whether there exists a pair $p,q$ such that $1\leq p,q \leq n - \ell + 1$ and $M_A[\ell][p] = M_B[\ell][q]$. Then return the highest value of $\ell$ and corresponding values of $p,q$.


In the binary case, a \pv ~is fully represented by just one arbitrary chosen
component. Hence, the set of \pv s of binary factors is just a one dimension
list of integers that can be stored in $\BO(n \log n)$ bits, since we have $n$ values in the range $[0\dd n]$.
The intersection can be accomplished in two steps. First, we sort the
$M_A[\ell]$ and $M_B[\ell]$ rows in $\BO(n)$ time by putting them in two lists
and using the classic Counting Sort algorithm \cite[Section 8.2]{Cormen01}.
Then, we check for a non empty intersection with a simple linear scan of the two
lists in linear time by starting in parallel from the beginning of the two lists
and moving forward element by element on the list having the smallest value
among the two examined elements. A further linear scan of $M_A[\ell]$ and $M_B[\ell]$ will find the indexes $p,q$ of an element of the not empty intersection. This gives us an $\BO(n^2)$ time algorithm requiring $\BO(n \log n)$ bits of space for computing an LCAF of two given binary strings.

In the more general case of alphabet greater than two, comparing two \pv s is no
more a constant time operation and checking for empty intersections is not a trivial task. In fact, sorting the set of vectors requires a full order to be defined. We can define an order component by component giving more value to the first component, then to the second one and so on. More formally, we define $x<y$, with $x,y \in \mathbb{N}^\sigma$, if there exist $1 \geq k \geq \sigma$ such that $x[k]<y[k]$ and, for any $i$ with $1\leq i < k$, $x[i]=y[i]$. Notice that comparing two vectors will take $BO(\sigma)$ time.

Now, one can sort two list of $n$ vectors of dimension $\sigma -1$, i.e.,
$M_A[\ell]$ and $M_B[\ell]$, in $\BO(\sigma ~n)$ by using $n$ comparisons taking
$\BO(\sigma)$ each. Therefore, now the algorithm runs in $\BO( \sigma ~n^2 )$
time using $\BO(\sigma ~n \log \sigma)$ bits of extra space.

\section{A Sub-quadratic Algorithm for the Binary Case}\label{SEC:subquadratic}
In Section \ref{SEC:quadratic}, we have presented an $\BO(n^2)$ algorithm to
compute the LCAF between two binary strings and two occurrences of common
abelian factors, one in each string, having LCAF length. In this section,
we show how we can achieve a better running time for the LCAF problem. We
will make use of the recent data structure of Moosa and Rahman
\cite{Moosa:2010:IPB:1837514.1837546} for indexing an abelian pattern. The
results of Moosa and Rahman \cite{Moosa:2010:IPB:1837514.1837546} is
presented in the form of following lemmas with appropriate rephrasing to
facilitate our description. 

\begin{lemma}\label{lem:interpolation}
(Interpolation lemma). If $S_1$ and $S_2$ are two substrings of a string
$S$ on a binary alphabet such that $\ell = |S_1| = |S_2|$, $i = |S_1|_1$,
$j = |S_2|_1$, $j > i + 1$, then, there exists another substring $S_3$
such that $\ell = |S_3|$ and $i < |S_3|_1 < j$.
\end{lemma}

\begin{lemma}\label{lem:n2bylogn} 
Suppose we are given a string $S$ of length $n$ on a binary alphabet.
Suppose that $maxOne(S, \ell)$ and $minOne(S, \ell)$ denote, respectively,
the maximum and minimum number of $1$'s in any substring of $S$ having
length $\ell$. Then, for all $1\leq \ell \leq n$, $maxOne(S, \ell)$ and
$minOne(S, \ell)$ can be computed in $\BO(n^2/\log n)$ time and linear space.
\end{lemma}

A result similar to Lemma \ref{lem:interpolation} is contained in the paper of Cicalese et al. \cite[Lemma 4]{DBLP:conf/stringology/CicaleseFL09}, while the result of Lemma \ref{lem:n2bylogn} has been discovered simultaneously and independently by Moosa and Rahman \cite{Moosa:2010:IPB:1837514.1837546} and by Burcsi et al. \cite{DBLP:conf/fun/BurcsiCFL10}.  
In addition to the above results we further use the following lemma.

\begin{lemma}\label{lem:l-factor}
Suppose we are given two binary strings $A, B$ of length $n$ each. There is a common abelian factor of $A$ and $B$ having length $\ell$ if and only if $maxOne(B, \ell)  \geq minOne(A, \ell)$ and $maxOne(A, \ell)  \geq   minOne(B, \ell)$.
\end{lemma}
\begin{proof}
Assume that $min_A = minOne(A, \ell)$, $max_A = maxOne(A, \ell)$, $min_B = \\minOne(B, \ell)$, $max_B = maxOne(B, \ell)$. Now by Lemma~\ref{lem:interpolation}, for all $min_A\leq k_A\leq max_A$, we have some $\ell$-length substrings $A(k_A)$ of $A$ such that $|A(k_A)|_1 = k_A$. Similarly, for all $min_B\leq k_B\leq max_B$, we have some $\ell$-length factors $B(k)$ of $B$ such that $|B(k_B)|_1 = k_B$. Now, consider the range $[min_A\dd max_A]$ and $[min_B\dd max_B]$. Clearly, these two ranges overlap if and only if $max_B \not< min_A$ and $max_A \not< min_B$. If these two ranges overlap then there exists some $k$ such that $min_A\leq k\leq max_A$ and $min_B\leq k\leq max_B$. Then we must have some substring $\ell$-length factors $A(k)$ and $B(k)$. Hence the result follows.
\end{proof}

Let us now focus on devising an algorithm for computing the LCAF given two
binary strings $A$ and $B$ of length $n$. For all $1\leq \ell \leq n$, we
compute $maxOne(A, \ell)$, $minOne(A, \ell)$, $maxOne(B, \ell)$ and
$minOne(B, \ell)$ in $\BO(n^2/ \log n)$ time (Lemma~\ref{lem:n2bylogn}).
Now we try to check the necessary and sufficient condition of
Lemma~\ref{lem:l-factor} for all $1\leq \ell \leq n$ starting from $n$ down
to $1$.
We compute the highest $\ell$ such that $$[minOne(A, \ell)\dd maxOne(A,
\ell)] \mbox{ and } [minOne(B, \ell)\dd maxOne(B, \ell)] \mbox{ overlap.}$$
Suppose that $\mathcal K$ is the set of values that is contained in the
above overlap, that is $\mathcal K = \{\  k~|~ k \in [minOne(A, \ell) \dd
maxOne(A, \ell)] ~\mathrm{and}~ k \in [minOne(B, \ell) \dd \\maxOne(B,
\ell)] \ \}$. Then by Lemma~\ref{lem:l-factor}, we must have a set
$\mathcal S$ of common abelian factors of $A, B$ such that for all $S\in
\mathcal S$, $|S| = \ell$. Since we identify the highest $\ell$, the length
of a longest common factor must be $\ell$, i.e., LCAF length is $\ell$.
Additionally, we have further identified the number of $1$'s in such
longest factors in the form of the set $\mathcal K$. Also, note that for a
$k \in \mathcal K$ we must have a factor $S\in \mathcal S$ such that $|S|_1
= k$.

Now let us focus on identifying an occurrence of the LCAF.
There are a number of ways to do that. But a straightforward and
conceptually easy way is to run the folklore $\ell$-window based algorithm
in \cite{Moosa:2010:IPB:1837514.1837546} on the strings $A$ and $B$ to find
the $\ell$-length factor with number of $1$'s equal to a particular value
$k\in \mathcal K$.

The overall running time of the algorithm is deduced as follows. By
Lemma~\ref{lem:n2bylogn}, the computation of $maxOne(A, \ell)$, $minOne(A,
\ell)$, $maxOne(B, \ell)$ and $minOne(B, \ell)$ can be done in $\BO(n^2/ \log
n)$ time and linear space. The checking of the condition of
Lemma~\ref{lem:l-factor} can be done in constant time for a particular
value of $\ell$. Therefore, in total, it can be done in $\BO(n)$ time.
Finally, the folklore algorithm requires $\BO(n)$ time to identify an
occurrence (or all of them)  of the factors. In total the running time is
$\BO(n^2/ \log n)$ and linear space.

\section{Towards a Better Time Complexity}\label{SEC:avg}

In this section we discuss a simple variant of the quadratic algorithm
presented in  \ref{SEC:quadratic}.
We recall that the main idea of the quadratic solution is to find the
greatest $\ell$ with $M_A[\ell] \bigcap M_B[\ell] \neq \emptyset$. The
variant we present here is based on the following two simple observations:
\begin{enumerate}
\item One can start considering sets of factors of decreasing lengths;
\item When an empty intersection is found between $M_A[\ell]$ and
$M_B[\ell]$, some rows can possibly be skipped based on the evaluation of
the \emph{gap} between $M_A[\ell]$ and $M_B[\ell]$.
\end{enumerate}

The first observation is trivial. The second observation is what we call
the \emph{skip trick}.
Assume that $M_A[\ell]$ and $M_B[\ell]$ have been computed and $M_A[\ell]
\bigcap M_B[\ell] = \emptyset$ have been found. It is easy to see that,
for any starting position $i$ and for any component $j$ (i.e., a letter
$a_j$), we have $$M_A[\ell][i][j]-1 \leq M_A[\ell-1][i][j] \leq
M_A[\ell][i][j]+1$$

Exploiting this property, we keep track, along the computation of
$M_A[\ell]$ and $M_B[\ell]$, of the minimum and maximum values that appear
in \pv s of factors of length $\ell$. We use four arrays indexed by
$\sigma$, namely $min_A, max_A,$ $min_B, max_B$. Notice that such arrays do
not represent \pv s as they just contain min and max values component by
component.
Formally, $min_A[j] = $ min$\{M_A[\ell][i][j]\}$, for any $i=1,\dots
\ell+1$. The others have similar definitions.

We compare, component by component, the range of $a_j$ in $A$ and $B$ and
we skip as many Rows as max$_{j=1}^{\sigma-1}(min_B[j] - max_A[j])$,
assuming $min_B[j] \geq max_A[j]$ (swap $A$ and $B$, otherwise). The
modified algorithm is reported in Algorithm 1.

Note that the tricks employed in our skip trick algorithm are motivated by
the fact that the expected value of the LCAF length of an independent and
identically distributed (i.i.d.) source is exponentially close to $n$
according to classic Large Deviation results \cite{ellisbook}. The same result
is classically extended to an ergodic source and it is meant to be a good
approximation for real life data when the two strings follow the same probability
distribution. Based on this, we have the following conjecture.

\begin{conjecture}
The expected length of LCAF between two strings $A,B$ drawn from an i.i.d.
source is LCAF$_{avg} = n - \BO(\log n)$, where $|A|=|B|=n$, and the number of
computed Rows in Algorithm \ref{algo:skip} is $\O(\log n)$ in average.
\end{conjecture}

Finally, we will make use of one more trick, that is, computing the
first vector of the current row in constant time from the first vector of the
previous row, when we skip some rows, instead of computing the new row from
scratch, we can use the first vector of the row below to compute the first
vector of the new row.
When we compute the rows we need, we will just populate the required
two lists and save a copy of the first vector of the computed row as we
will need it along the next iterative steps as shown in Algorithm
\ref{algo:FVtrick}.

For instance, if we know $\mathcal M[\ell]$ and we jump to $\mathcal
M[\ell-3]$, i.e., we skip $\mathcal M[\ell-1]$ and $\mathcal M[\ell-2]$,
we take $\mathcal M[\ell][1]$ and compute in constant time $\mathcal M[\ell-1][1]$,
$\mathcal M[\ell-2][1]$, then again compute $\mathcal M[\ell-3][1]$.
From  $\mathcal M[\ell][1]$, to compute $\mathcal M[\ell-1][1]$, we have to
subtract $1$ from the vector $\mathcal M[\ell][1]$ at index $\s{s}[\ell]$,
that is the last character of the factor of length $\ell$ starting at
position $1$ (i.e., $\mathcal M[\ell][1]$). For example, consider 
$\s{s}=aacgcctaatcg$, we have $\mathcal M[12][1] = (4a,4c,2g,2t)$ and
$\mathcal M[11][1]=(4a,4c,1g,2t)$, i.e., $(4a,4c,2g,2t)$ minus $1g$.

 \begin{algorithm}[H]
\caption{Compute $\LCAF$ of $\s{x}$ and $\s{y}$ using the \itbf{skip trick}.}
\label{algo:skip}
\begin{algorithmic}[1]
\vspace{5pt}
\Function {ComputeLCAF}{$\s{x},\s{y}$}
	\State set $\ell = n = |\s{x}|$
	\State set found $=$ False

	\State compute $max_{\s{x}} = $  MAX($\s{x},\ell$),
	 $max_{\s{y}} = $  MAX($\s{y},\ell$) 
	\State compute $min_{\s{x}} = $  MIN($\s{x},\ell$),
	 $min_{\s{y}} = $  MIN($\s{y},\ell$)

	\If {$max_{\s{x}} == max_{\s{y}}$}
		\State found $=$ True
	\Else
		\State $\ell = \ell - $ SKIP($min_{\s{x}}, max_{\s{x}}, min_{\s{y}}, max_{\s{y}}$)
	\EndIf
	
	\While{(found $==$ False) AND ($\ell \geq 0$)}
		\State compute $max_{\s{x}} =$  MAX($\s{x},\ell$), $min_{\s{x}}  = $ 
		MIN($\s{x},\ell$) 
		\State compute $max_{\s{y}}  =$  MAX($\s{y},\ell$),
		$min_{\s{y}}   =$ MIN($\s{y},\ell$)

		\State compute list$_{\s{x}} = \mathcal M_{\s{x}}[\ell]$,
		 list$_{\s{y}} = \mathcal M_{\s{y}}[\ell]$
		\State sort list$_{\s{x}}$, list$_{\s{y}}$
		\State compute list$_{\s{x}} \bigcap $ list$_{\s{y}}$
		\If{list$_{\s{x}} \bigcap $ list$_{\s{y}} \neq \emptyset$}
			\State found $=$ True
			\State break
		\EndIf
		\State $\ell = \ell - $ SKIP($min_{\s{x}}, max_{\s{x}}, min_{\s{y}}, max_{\s{y}}$)
	\EndWhile
	\State return $\ell$
\EndFunction
\vspace{3pt}
\Function {max}{$\s{s},\ell$}
	\State int count$[\sigma]$,max$[\sigma]$

	\For {($i=1$; $i\leq \ell$; $i$++)}
		\State count$[\s{s}[i]]$++
	\EndFor

	\State max $=$ count
	
	\For {($i=\ell$; $i<|\s{s}|-\ell$; $i$++)} 
		\State count$[\s{s}[i-1]]$- -
		\State count$[\s{s}[i+\ell-1]]$++
		\If {count$[\s{s}[i+\ell-1]] \geq $ max$[\s{s}[i+\ell-1]]$}
			\State max$[\s{s}[i+\ell-1]] =$ count$[\s{s}[i+\ell-1]]$
		\EndIf
	\EndFor
	\State return max
\EndFunction
\vspace{3pt}
\Function {skip}{$min_{\s{x}}, max_{\s{x}}, min_{\s{y}}, max_{\s{y}}$}
	\State int gap$[\sigma-1]$
	\For {($j=1$; $j<\sigma$; $j$++)}
		\If {$max_{\s{x}}[j] >= min_{\s{y}}[j]$}
			\State  gap$[j] = |min_{\s{x}}[j] - max_{\s{y}}[j]|$
		\Else
			\State gap$[j] = |min_{\s{y}}[j] - max_{\s{x}}[j]|$
		\EndIf
	\EndFor
	\State return max(gap)
\EndFunction
\end{algorithmic}
\end{algorithm}
   
\section{Experiments}\label{sec:exp}
We have conducted some experiments to analyze the behaviour and running
time of our skip trick algorithm in practice. The experiments have been run
on a Windows Server 2008 R2 64-bit Operating System, with Intel(R) Core(TM)
i7 2600 processor @ 3.40GHz having an installed memory (RAM) of 8.00 GB.
Codes were implemented in $C\#$ language using Visual Studio 2010.

\begin{figure}[h!]
\centering 
\includegraphics*[scale=0.47]{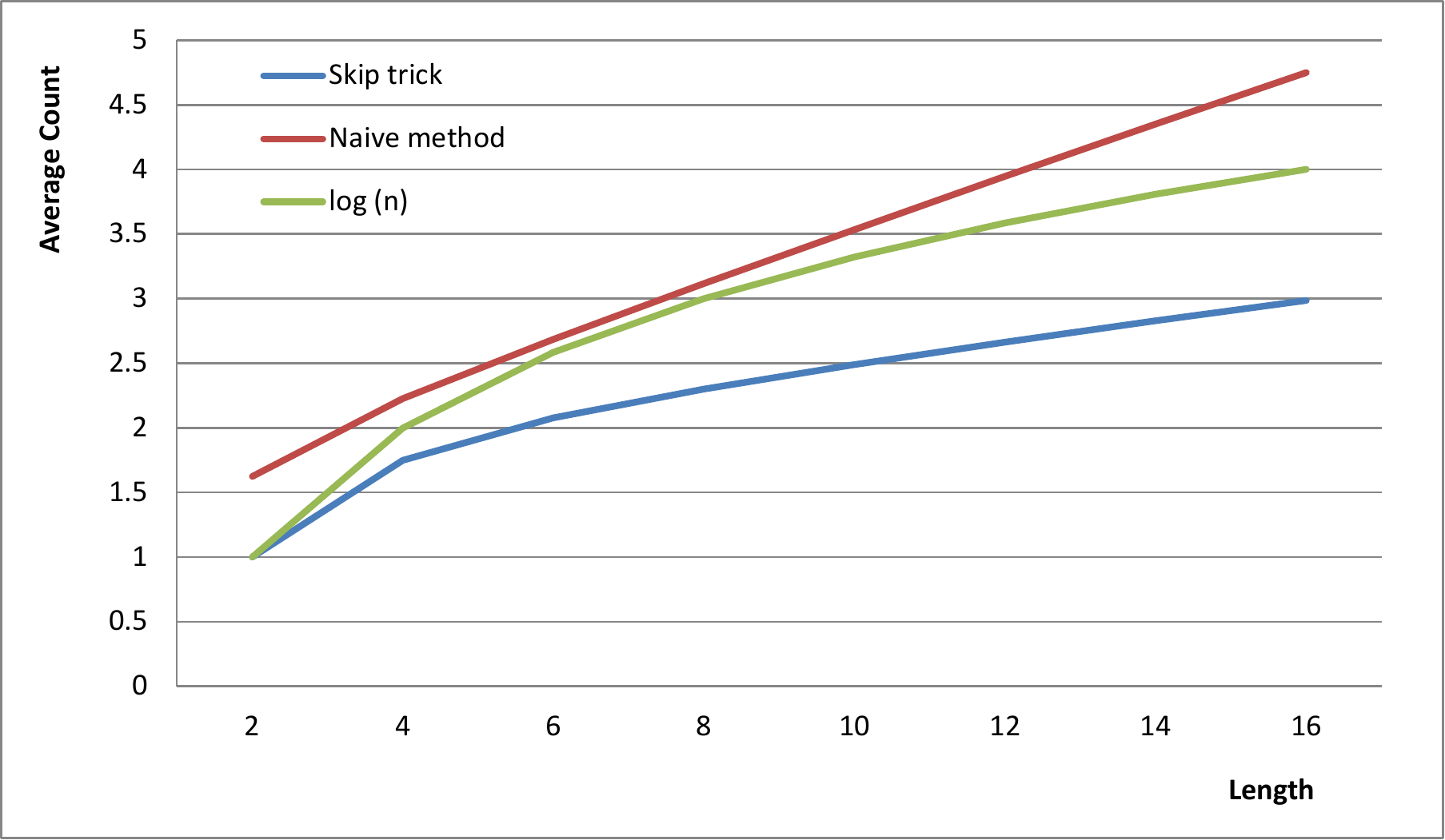}
\caption{Plot of the average number of rows computed executing Algorithm
\ref{algo:skip} on all the strings of length $2, 3, \dots 16$ over the
binary alphabet.}
\label{fig:avgstat1}
\end{figure}

Our first experiment have been carried out principally to verify our
rationale behind using the skip trick.
We experimentally evaluated the expected number of rows computed in average
by using the skip trick of Algorithm \ref{algo:skip}.

\begin{figure}[h!]
\centering
\includegraphics*[scale=0.47]{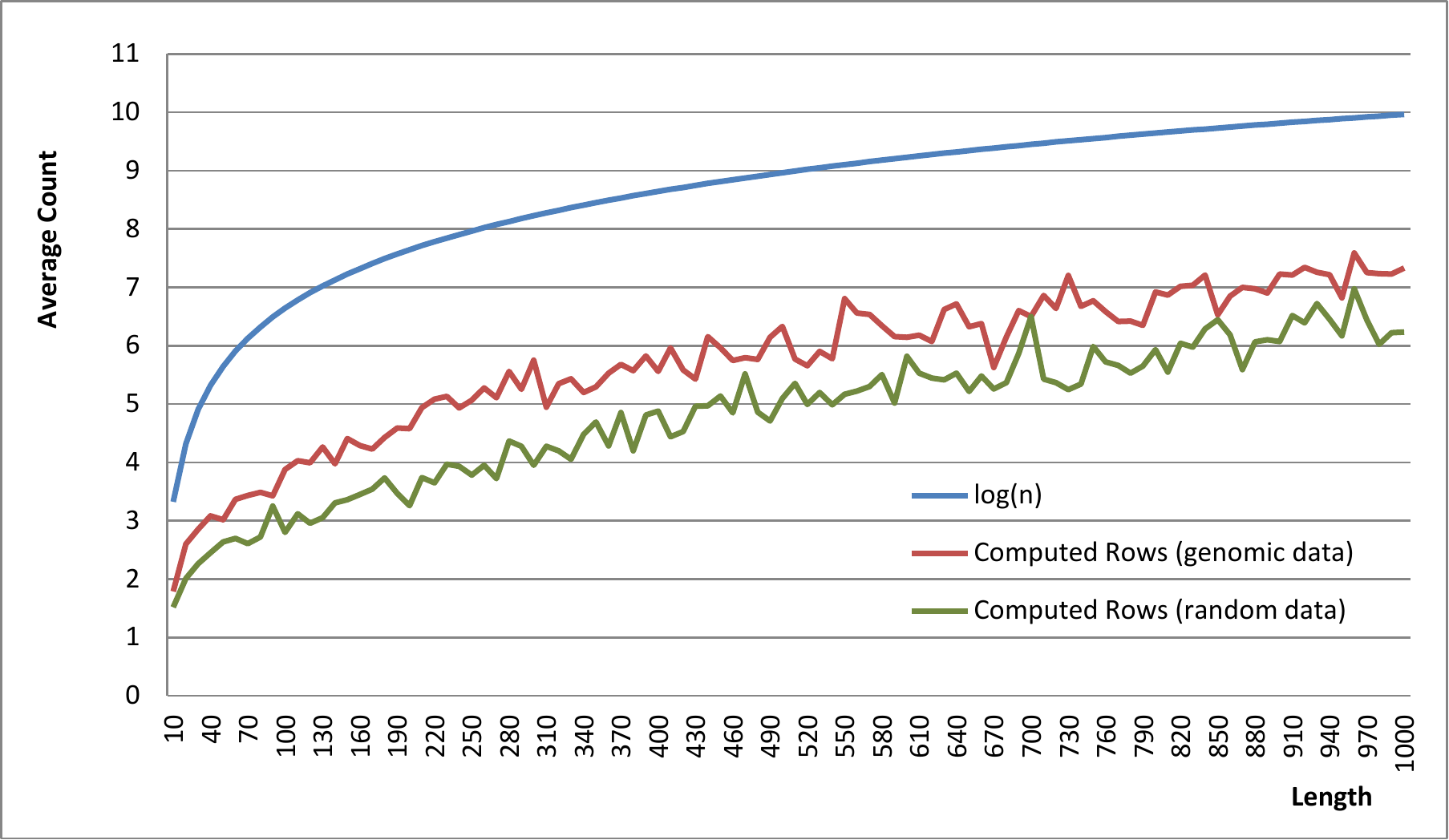}
\caption{Plot of the average number of rows computed executing Algorithm
\ref{algo:skip} on both genomic and random datasets over the DNA alphabet.}
\label{fig:avgstat2}
\end{figure}

Figure \ref{fig:avgstat1} shows the average number of rows computed
executing Algorithm \ref{algo:skip} on all the strings of length $2, 3,
\dots 16$ over the binary alphabet. Naive method line refers to the number
of rows used without the skip trick, but starting from $\ell=n$ and
decreasing $\ell$ by one at each step. Notice that the skip trick line is
always below the $\log n$ line.

To this end we have conducted an experiment to evaluate the expected number
of rows computed by our skip trick algorithm. In particular, we have
implemented the skip trick algorithm as well as the naive algorithm and
have counted the average number of rows computed by the algorithms on all
the strings of length $2, 3, \dots 16$ on binary alphabet. The results are
reported in Figure \ref{fig:avgstat1}. It shows that the computed rows of
$\s{x},\s{y}$, starting from $\ell=n$ to $\ell=n- \log n$, sum up to $\BO(\log
n)$.
 
On the other hand, to reach a conclusion in this aspect we
would have to increase the value of $n$ in our experiment to substantially
more than 64; for $n=64$, $\sqrt n$ is just above $\log n$. Regrettably,
limitation of computing power prevents us from doing such an experiment. So, we resort
to two more (non-comprehensive) experimental setup as follows to check the
practical running time of the skip trick algorithm.

\begin{figure}[h!]
\centering
  \includegraphics*[scale=0.47]{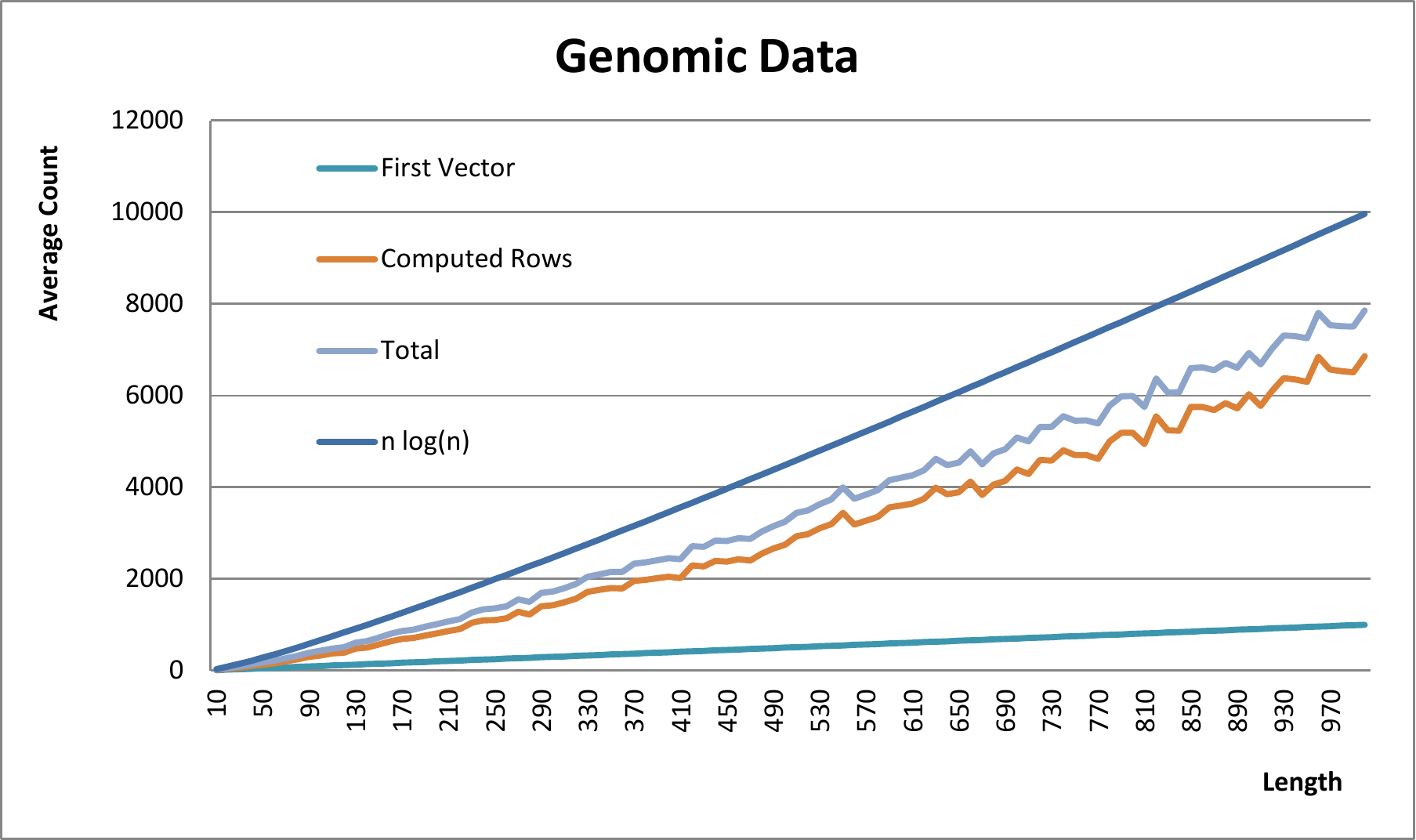}
\caption{Plot of the average number of rows computed executing Algorithm 1 on
sequences taken from the Homo sapiens genome.}
\label{fig:genomic}
\end{figure}

Furthermore, we conduct our experiments on two datasets, real genomic
data and random data. We have taken a sequence ($\mathcal S$) from the Homo
sapiens genome (250MB) for the former dataset. The latter dataset is
generated randomly on the DNA alphabet (i.e., $\Sigma=\{a, c, g, t\}$). In
particular, Here we have run the skip trick algorithm on 2 sets of pairs
of strings of lengths $10, 20, .., 1000$. For the genomic dataset, these
pairs of strings have been created as follows.  For each length $\ell, \ell
\in \{10, 20, .., 1000\}$ two indexes $i, j \in [1..|\s{x}|-\ell]$ have
been randomly selected to get a pair of strings $\mathcal S[i..i+\ell-1],
\mathcal S[j..i+\ell-1]$, each of length $\ell$. A total of 1000 pairs of
strings have been generated in this way for each length $\ell$ and the skip
trick algorithm has been run on these pairs to get the average results. On
the other hand for random dataset, we simply generate the same number of
strings pairs randomly and run the skip trick algorithm on each pair of
strings and get the average results for each length group. In both cases,
we basically count the numbers of computed rows.

\begin{figure}[h!]
\centering
  \includegraphics*[scale=0.47]{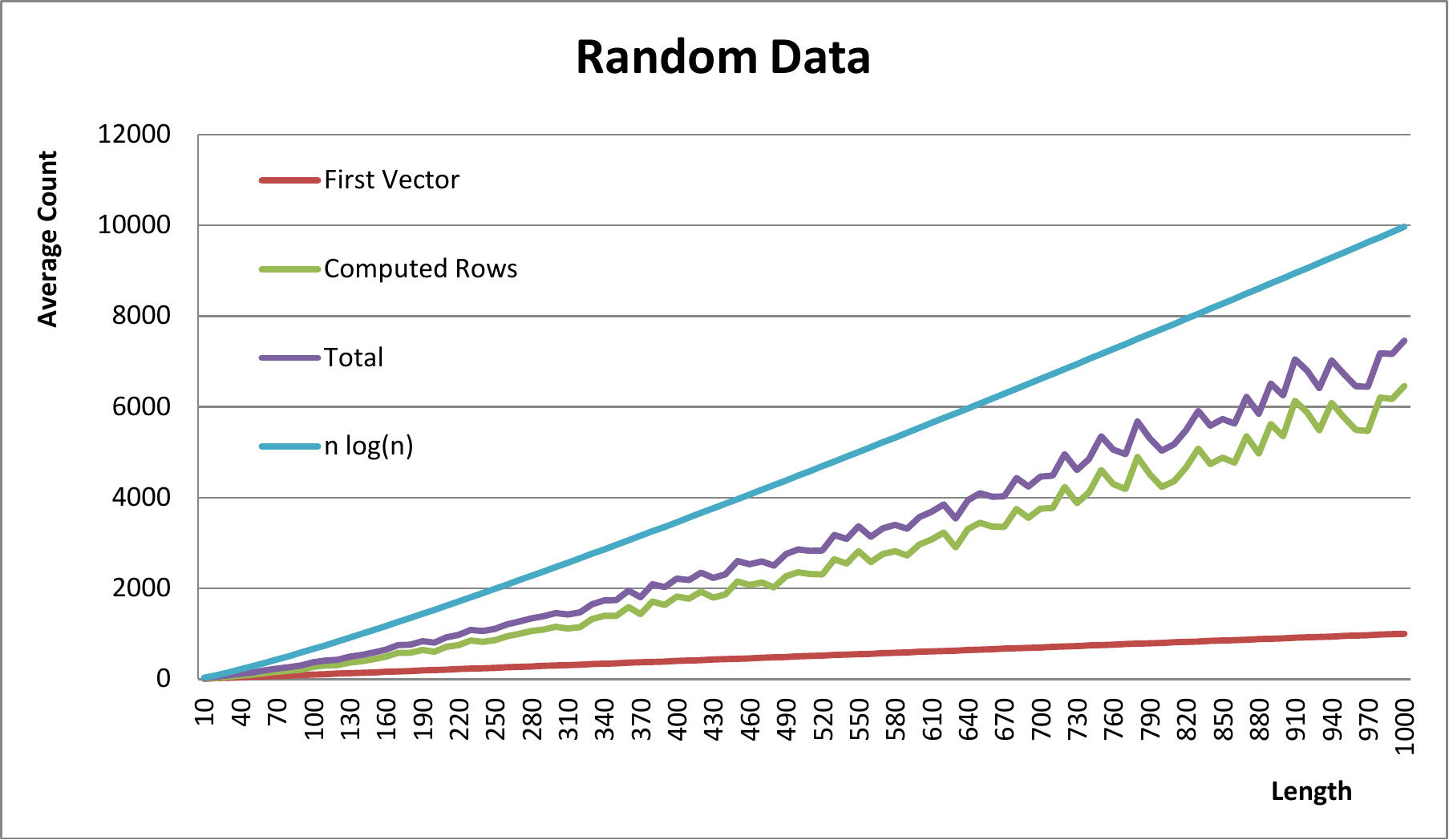}
\caption{Plot of the average number of rows computed executing Algorithm 1 on
randomly generated sequences over the alphabet $\Sigma=\{a, c, g, t\}$.}
\label{fig:random}
\end{figure}

Figure \ref{fig:avgstat2} shows the average number of rows computed
executing Algorithm \ref{algo:skip} on both genomic and random datasets
over the DNA alphabet (i.e., $\Sigma=\{a, c, g, t\}$). Notice that the skip
trick line is always below the $\log n$ line.
Figure \ref{fig:avgstat2} shows that the computed rows of $\s{x},\s{y}$,
starting from $\ell=n$ to $\ell=n - \log n$, sum up to $\BO(\log n)$.

We experimentally evaluated the computing of the first
vector and the expected number of rows computed in average by employing the
first vector trick (Algorithm \ref{algo:FVtrick}). We have used the same
experiment configuration as the above. The average number of rows and of
the first vector computed executing Algorithm \ref{algo:FVtrick} on both
genomic and random datasets over the DNA alphabet (i.e., $\Sigma=\{a, c,
g, t\}$). In both cases, we basically count the numbers of computed rows
and the first vector. The results are illustrated in Figures
\ref{fig:genomic} and \ref{fig:random}.

In both cases, The figures report the average count of computed rows
(Number of Rows), the average count of the first vector (First Vector) and
the summation of these two counts (Total). It also shows the $n\log n$
curve. Both of the figures show that the algorithm computed the first
vector of the visited rows in $\BO(n)$ and the total running time for
Algorithm \ref{algo:FVtrick} would be $\BO(n\log n)$ in practice. 

Since any row computation takes $\BO(\sigma ~n)$, this suggests an average
time complexity of $\BO(\sigma~ n~ \log n)$, i.e., $\BO(n~ \log n)$ for a
constant alphabet.

\section{Conclusion}\label{SEC:conclusion}
In this paper we present a simple quadratic running time algorithm for the LCAF
problem and  a sub-quadratic running time solution for the binary string case,
both having linear space requirement. Furthermore, we present a variant of the
quadratic solution that is experimentally shown to achieve a better time
complexity of $\BO(n \log n)$.

\begin{algorithm}[H]
\begin{algorithmic}[1]
\Function {first}{$\s{s},\ell$}
    \State int first$[\sigma]$
    \For {($i=1$; $i<\ell$; $i$++)}
        \State first$[\s{s}[i]]$++
	\EndFor
	\State return first 
\EndFunction 
\vspace{3pt}
\Function {row}{$\s{s}$, $\ell$, first}
    \State int row$[\sigma]$
    \State   row =first
    \For {($i=1$; $i<|\s{s}| - \ell$; $i$++)}
        \State row$[\s{s}[i - 1]]$- -
        \State row$[\s{s}[i + l - 1]]$++
	\EndFor
	\State return row
\EndFunction
\vspace{3pt}
\Function {ComputeLCAF}{$\s{x},\s{y}$}
	\State set $\ell = n = |\s{x}|$
	\State set found $=$ False

    \State compute $first_{\s{x}} =  FIRST(\s{x},\ell)$
	\State compute $first_{\s{y}} =  FIRST(\s{y},\ell)$

	
	\While{(found $==$ False) AND ($\ell \geq 0$)}

        \State compute $row_{\s{x}} =  ROW(\s{x},\ell, first_{\s{x}})$
        \State compute $row_{\s{y}} =  ROW(\s{y},\ell, first_{\s{y}})$

		\State compute list$_{\s{x}} = \mathcal M_{\s{x}}[\ell]$, list$_{\s{y}} = \mathcal M_{\s{y}}[\ell]$
		\State sort list$_{\s{x}}$, list$_B$
		\State compute list$_{\s{x}} \bigcap $ list$_{\s{y}}$
		\If{list$_{\s{x}} \bigcap $ list$_{\s{y}} \neq \emptyset$}
			\State found $=$ True
			\State break
		\EndIf 

        \State compute $max_{\s{x}} =  MAX(\s{x},\ell)$, $min_{\s{x}}  =  MIN(\s{x},\ell)$
        \State compute $max_{\s{y}}  =  MAX(\s{y},\ell)$, $min_{\s{y}} =  MIN(\s{y},\ell)$

		\State $\ell = \ell - $ SKIP($min_{\s{x}}, max_{\s{x}}, min_{\s{y}}, max_{\s{y}}$)
	\EndWhile
	\State return $\ell$
\EndFunction
\end{algorithmic}
\caption{Compute $\LCAF$ of $\s{x}$ and $\s{y}$ using the \itbf{first
vector trick}.}
\label{algo:FVtrick}
\end{algorithm}
 
\section*{Acknowledgement}
We thank Thierry Lecroq and Arnaud Lefebvre for proposing the LCAF problem and
the participants of Stringmasters 2013 meetings for helping us to get more
involved with this topic.

\bibliographystyle{alpha}
\bibliography{lcafr}

\end{document}